\documentclass[a4paper, 10pt]{extarticle}
\usepackage{amsmath}
\usepackage{amssymb}
\usepackage{amsfonts}
\usepackage{graphicx}
\usepackage{subfigure}
\usepackage{authblk}
\usepackage[textwidth=16cm,textheight=20cm]{geometry}
\newtheorem{theorem}{Theorem}

\newtheorem{example}[theorem]{Example}

\newtheorem{lemma}[theorem]{Lemma}

\newenvironment{proof}[1][Proof]{\noindent\textbf{#1.} }{\ \rule{0.5em}{0.5em}}
\title{Schr\"{o}dinger equations on elliptic curves: symmetries, solutions and eigenvalue problem}
\author[1,\thanks{\textit{E-mail: }\texttt{valentin.lychagin@uit.no}}]{Valentin Lychagin}
\author[1,\thanks{\textit{E-mail: }\texttt{mihail\underline{ }roop@mail.ru}}]{Mikhail Roop}
\affil[1]{V.A. Trapeznikov Institute of Control Sciences, Russian Academy of Sciences, 65 Profsoyuznaya Str., 117997 Moscow, Russia}

\begin{document}
\maketitle

\abstract{
In this paper, we study Schr\"{o}dinger equations on elliptic curves called generalized Lam\'{e} equations. We suggest a method of finding integrable potentials for Schr\"{o}dinger type equations. We apply this method to the Lam\'{e} equations and provide a sequence of integrable potentials for which the eigenvalue problem is solved explicitly.
}

\section{Introduction}
The Lam\'{e} equation firstly appeared in \cite{Lame} in separation of variables for the Laplace equation in elliptic coordinates. Later the Lam\'{e} equation  was used in various problems of quantum mechanics, for example in theory of periodic instantons \cite{LMKT} and also appeared as Schr\"{o}dinger equation for periodic potentials (see, for example, \cite{MK}). The present work is devoted to generalized Lam\'{e} equations and provides a method of finding analytical solutions to them, as well as solutions to the eigenvalue problem.

This paper has the following structure. In Sect. 2, we describe a method of finding integrable potentials and solutions for Schr\"{o}dinger type equations. We recall necessary constructions from the geometrical theory of ODEs and apply its results to linear second order ODEs of Schr\"{o}dinger type. We show that having a symmetry one can obtain an infinite sequence of integrable potentials and corresponding solutions. Sect. 3 is devoted to Schr\"{o}dinger equations on elliptic curves, the so-called generalized Lam\'{e} equations. Such curves are parameterized by Weierstrass $p$-functions and we analyze the case when both potential and symmetry are polynomials in Weierstrass $p$-function and linear with respect to its first derivative. The most complete results are obtained for so-called even and odd cases, when the potential is an even function and the symmetry function is either even, or odd. We get a series of integrable potentials linear in the Weierstrass $p$-function, for which one can solve both eigenvalue problem and corresponding Schr\"{o}dinger equation explicitly.

The results of this paper can be interpreted in the light of more general problem: how can one construct solutions of Schr\"{o}dinger equations with potential satisfying a given ODE? The present paper provides the method for potentials satisfying $n$-th stationary KdV equations \cite{LychLych}, and other cases are the subject of further elaboration.

\section{Symmetries and integrals for Schr\"{o}dinger type equations}
\subsection{Symmetries of ODEs}
Here, we briefly describe a geometrical approach to ordinary differential equations following \cite{VinKr,KLR}.

Let us consider an ordinary differential equation of order $k$. We will restrict our consideration to resolved with respect to the highest derivative ODEs of the form
\begin{equation*}
u^{(k)}=f\left(x,u,u^{\prime},\ldots u^{(k-1)}\right),
\end{equation*}
which can be naturally associated with a smooth submanifold $\mathcal{E}$ in a space of $k$-jets $J^{k}(\mathbb{R})$ of functions of $x$ with canonical coordinates $(x,u_{0},\ldots,u_{k})$ corresponding to the independent variable, unknown function and its derivatives up to order $k$:
\begin{equation*}
\mathcal{E}=\left\{u_{k}=f(x,u_{0},\ldots,u_{k-1})\right\}\subset J^{k}(\mathbb{R}).
\end{equation*}
The space of $k$-jets $J^{k}(\mathbb{R})$ is equipped with the Cartan distribution
\begin{equation*}
\mathcal{C}\colon J^{k}\ni\theta\mapsto\mathcal{C}_{\theta}\subset T_{\theta}J^{k},
\end{equation*}
generated by vector fields $\partial_{u_{k}}$ and
\begin{equation*}
\mathcal{D}=\partial_{x}+u_{1}\partial_{u_{0}}+\cdots+u_{k}\partial_{u_{k-1}},
\end{equation*}
or, equivalently, by the Cartan forms
\begin{equation*}
\omega_{j}=du_{j}-u_{j+1}dx,\quad j=\overline{0,k-1}.
\end{equation*}
The restriction $\mathcal{C}_{\mathcal{E}}$ of the Cartan distribution to the submanifold $\mathcal{E}$ is a one-dimensional distribution almost everywhere on $\mathcal{E}$:
\begin{equation*}
\mathcal{C}_{\mathcal{E}}\colon\mathcal{E}\ni\theta\mapsto\mathcal{C}_{\mathcal{E}}(\theta)=T_{\theta}\mathcal{E}\cap\mathcal{C}_{\theta},
\end{equation*}
except at points $\theta\in\mathcal{E}$ where $\mathcal{C}_{\theta}\subset T_{\theta}\mathcal{E}$, which are called \textit{singular}. The distribution $\mathcal{C}_{\mathcal{E}}$ can therefore be given by a vector field
\begin{equation*}
X_{f}=\mathcal{D}+\mathcal{D}(f)\partial_{u_{k}},
\end{equation*}
and an integral curve $l\subset\mathcal{E}$ of the distribution $\mathcal{C}_{\mathcal{E}}$ is said to be a solution of the equation $\mathcal{E}$.

A transformation $\Phi\colon\mathcal{E}\to\mathcal{E}$ is called \textit{symmetry} of the equation $\mathcal{E}$ if it preserves the Cartan distribution $\mathcal{C}_{\mathcal{E}}$, i.e. $\Phi_{*}(\mathcal{C}_{\mathcal{E}})=\mathcal{C}_{\mathcal{E}}$. Infinitesimally, such a transformation is generated by a vector field $X\in D(\mathcal{E})$, such that $[X,\mathcal{C}_{\mathcal{E}}]\subset\mathcal{C}_{\mathcal{E}}$. Infinitesimal symmetries form a Lie algebra $\mathrm{Sym}(\mathcal{C}_{\mathcal{E}})$ with respect to the commutator of vector fields. Obviously, vector fields from the distribution $\mathcal{C}_{\mathcal{E}}$ are infinitesimal symmetries themselves. They are called \textit{trivial} or \textit{characteristic} since they transform any solution to the same solution and therefore do not give us new solutions. They form an ideal $\mathrm{Char}(\mathcal{C}_{\mathcal{E}})$ of the Lie algebra $\mathrm{Sym}(\mathcal{C}_{\mathcal{E}})$ and elements of a quotient algebra $\mathrm{Shuff}(\mathcal{C}_{\mathcal{E}})=\mathrm{Sym}(\mathcal{C}_{\mathcal{E}})/\mathrm{Char}(\mathcal{C}_{\mathcal{E}})$ are called \textit{shuffle symmetries}. Shuffle symmetries $X\in\mathrm{Shuff}(\mathcal{C}_{\mathcal{E}})$ are defined by means of \textit{generating functions}~\cite{VinKr,KLR}.
\begin{theorem}
Shuffle symmetries of the ODE $\mathcal{E}$ are of the form
\begin{equation*}
X_{\phi}=\sum\limits_{j=0}^{k-1}\overline{\mathcal{D}}^{j}(\phi)\partial_{u_{j}},
\end{equation*}
where $\overline{\mathcal{D}}=\partial_{x}+u_{1}\partial_{u_{0}}+\cdots+f\partial_{u_{k-1}}$ is an operator of a total derivative on $\mathcal{E}$ and $\phi\in C^{\infty}(\mathcal{E})$ is a generating function.

The generating function is found from the Lie equation
\begin{equation}
\label{LieEq}
\overline{\mathcal{D}}^{k}(\phi)-\sum\limits_{j=0}^{k-1}f_{u_{j}}\overline{\mathcal{D}}^{j}(\phi)=0.
\end{equation}
\end{theorem}

From now and on by a symmetry of the equation $\mathcal{E}$ we shall mean a generating function $\phi$.

The Lie algebra structure in $\mathrm{Shuff}(\mathcal{C}_{\mathcal{E}})$ induces a Lie algebra structure on a space of generating functions by the following way:
\begin{equation*}
X_{[\phi_{1},\phi_{2}]}=[X_{\phi_{1}},X_{\phi_{2}}],
\end{equation*}
and an explicit expression for the bracket is
\begin{equation*}
[\phi_{1},\phi_{2}]=X_{\phi_{1}}(\phi_{2})-X_{\phi_{2}}(\phi_{1}).
\end{equation*}

\subsection{Schr\"{o}dinger equations}
Consider an ODE of the form
\begin{equation}
\label{SchEq}
y^{\prime\prime}+w(x)y=0,
\end{equation}
where $y(x)$ is an unknown function and $w(x)$ is a potential. Equation (\ref{SchEq}) defines a smooth submanifold
\begin{equation}
\label{eqE}
\mathcal{E}=\left\{u_{2}=-w(x)u_{0}\right\}\subset J^{2}(x,u_{0},u_{1},u_{2}).
\end{equation}
We will be interested in linear symmetries of (\ref{eqE}):
\begin{equation}
\label{sym}
\phi=a(x)u_{0}+b(x)u_{1},
\end{equation}
where $a(x)$ and $b(x)$ are some functions. Substituting (\ref{sym}) into (\ref{LieEq}) we get (see also~\cite{LychLych})
\begin{equation*}
\phi=cu_{0}+\phi_{z},
\end{equation*}
where $c$ is a constant,
\begin{equation*}
\phi_{z}=z(x)u_{1}-\frac{z^{\prime}(x)u_{0}}{2},
\end{equation*}
and the function $z(x)$ satisfies the \textit{Lie equation}:
\begin{equation}
\label{lieZ}
z^{\prime\prime\prime}+4wz^{\prime}+2w^{\prime}z=0.
\end{equation}
Thus equation (\ref{eqE}) has two commuting symmetries $\phi_{1}=u_{0}$ and $\phi_{2}=\phi_{z}$ and having a solution of (\ref{lieZ}) for a given potential $w(x)$ one can therefore integrate (\ref{eqE}) using the Lie-Bianchi theorem~\cite{VinKr,KLR}. We will call such potentials \textit{integrable}.

Introduce the operators
\begin{equation*}
S_{w}=\partial^{2}+w,\quad L_{w}=\partial^{3}+4w\partial+2w^{\prime}
\end{equation*}
corresponding to equations (\ref{SchEq}) and (\ref{lieZ}) respectively, and let $\mathrm{Sol}(w)$ and $\mathrm{L}(w)$ be solution spaces of Schr\"{o}dinger equation (\ref{SchEq}) and Lie equation (\ref{lieZ}) respectively, i.e.
\begin{eqnarray*}
\mathrm{Sol}(w)&=&\left\{y\mid S_{w}(y)=0\right\}\\
\mathrm{L}(w)&=&\left\{z\mid L_{w}(z)=0\right\}.
\end{eqnarray*}
Note that there is a correspondence between $\mathrm{Sol}(w)$ and $\mathrm{L}(w)$. Namely, if $y\in\mathrm{Sol}(w)$, then $z=y^{2}\in\mathrm{L}(w)$ and $\mathrm{L}(w)$ is therefore a symmetric square of $\mathrm{Sol}(w)$, $\mathrm{L}(w)=\mathbf{S}^{2}(\mathrm{Sol}(w))$. Moreover, $\mathrm{L}(w)=\mathfrak{sl}_{2}(\mathbb{R})$ with a bracket
\begin{equation*}
[z_{1},z_{2}]=z_{1}^{\prime}z_{2}-z_{1}z_{2}^{\prime}.
\end{equation*}
Indeed, if $y_{1},\,y_{2}\in\mathrm{Sol}(w)$ is a fundamental solution of (\ref{SchEq}), then solutions $A=y_{1}^{2}$, $B=y_{2}^{2}$, $C=2y_{1}y_{2}$ of (\ref{lieZ}) satisfy $\mathfrak{sl}_{2}(\mathbb{R})$ structure equations:
\begin{equation*}
[A,B]=C,\quad [C,A]=-2A,\quad [C,B]=2B.
\end{equation*}

Let us consider equation (\ref{lieZ}) as an equation for $w(x)$. The following theorem is valid.
\begin{theorem}
The symmetry function $z(x)$ and potential $w(x)$ are related as
\begin{equation}
\label{potsym}
w(x)=\frac{c_{w}}{z^{2}}+\frac{1}{4}\left(\frac{z^{\prime}}{z}\right)^{2}-\frac{z^{\prime\prime}}{2z},
\end{equation}
where $c_{w}$ is a constant,
\begin{equation*}
c_{w}=\frac{1}{4}K(z,z),
\end{equation*}
where $K$ is the Killing form of the Lie algebra $\mathrm{L}(w)$.

Let $\widehat{w}(x)$ be another potential with the same symmetry $z(x)$. Then,
\begin{equation*}
\widehat{w}=w+\frac{\widehat{c}}{z^{2}},
\end{equation*}
where $\widehat{c}=c_{\widehat{w}}-c_{w}$ is a constant.
\end{theorem}

Let us now get solutions to (\ref{SchEq}) having known the symmetry $z(x)$.
\begin{lemma}
Function $H=\phi_{1}\overline{\mathcal{D}}(\phi_{2})-\phi_{2}\overline{\mathcal{D}}(\phi_{1})$ is the first integral of (\ref{SchEq}) for any symmetries $\phi_{1}$ and $\phi_{2}$.
\end{lemma}
\begin{proof}
Since $\phi_{1}$ and $\phi_{2}$ are symmetries, $\overline{\mathcal{D}}^{2}(\phi_{1,2})=-w(x)\phi_{1,2}$ due to (\ref{LieEq}).
\begin{equation*}
\overline{\mathcal{D}}(H)=\phi_{1}\overline{\mathcal{D}}^{2}(\phi_{2})-\phi_{2}\overline{\mathcal{D}}^{2}(\phi_{1})=0.
\end{equation*}
\end{proof}

Applying the result of the above lemma to $\phi_{1}=\phi_{z}$ and $\phi_{2}=u_{0}$, we get
\begin{equation*}
H=\frac{c_{w}}{z}(u_{0})^{2}+\frac{1}{z}(\phi_{z})^{2}.
\end{equation*}
Introducing a new variable $v=u_{0}/\sqrt{|z|}$, we get $(v^{\prime})^{2}=z^{-3}(\phi_{z})^{2}$ and
\begin{equation*}
H=c_{w}v^{2}+z^{2}\left(v^{\prime}\right)^{2}=H_{0}^{2}
\end{equation*}
for some constant $H_{0}>0$. Consider three cases.
\begin{itemize}
\item Elliptic case, $c_{w}=q_{0}^{2}>0$.

Introduce a new variable $\psi$ by the following way:
\begin{equation*}
v=\frac{H_{0}}{q_{0}}\sin\psi,\quad v^{\prime}=\frac{H_{0}}{z}\cos\psi.
\end{equation*}
The last implies that
\begin{equation*}
\psi=\int\frac{q_{0}}{z}dx,\quad y=\frac{H_{0}}{q_{0}}\sqrt{|z|}\sin\left(q_{0}\int\frac{dx}{z}\right).
\end{equation*}
\item Hyperbolic case, $c_{w}=-q_{0}^{2}<0$

In the same way we obtain
\begin{equation*}
y=\frac{H_{0}}{q_{0}}\sqrt{|z|}\sinh\left(q_{0}\int\frac{dx}{z}\right).
\end{equation*}
\item Parabolic case, $c_{w}=0$
\begin{equation*}
y=H_{0}\sqrt{|z|}\int\frac{dx}{z}.
\end{equation*}
\end{itemize}
Summarizing above discussion, we have the following theorem.
\begin{theorem}
\label{thm1}
Let $z(x)$ be a nonzero symmetry of (\ref{SchEq}). Then, a fundamental solution of (\ref{SchEq}) is given as
\begin{itemize}
\item for $c_{w}=q_{0}^{2}>0$
\begin{equation}
\label{FundSolEll}
y^{(1)}(x)=\sqrt{|z|}\sin\left(q_{0}\int\frac{dx}{z}\right),\quad y^{(2)}(x)=\sqrt{|z|}\cos\left(q_{0}\int\frac{dx}{z}\right).
\end{equation}
\item for $c_{w}=-q_{0}^{2}<0$
\begin{equation}
\label{FundSolHyp}
y^{(1)}(x)=\sqrt{|z|}\sinh\left(q_{0}\int\frac{dx}{z}\right),\quad y^{(2)}(x)=\sqrt{|z|}\cosh\left(q_{0}\int\frac{dx}{z}\right).
\end{equation}
\item for $c_{w}=0$
\begin{equation}
\label{FundSolPar}
y^{(1)}(x)=\sqrt{|z|}\int\frac{dx}{z},\quad y^{(2)}(x)=\sqrt{|z|}.
\end{equation}
\end{itemize}
A fundamental solution of (\ref{lieZ}) is given as
\begin{itemize}
\item for $c_{w}=q_{0}^{2}>0$
\begin{equation}
\label{zFundSolEll}
z^{(1)}(x)=z,\, z^{(2)}(x)=z\sin\left(2q_{0}\int\frac{dx}{z}\right),\, z^{(3)}(x)=z\cos\left(2q_{0}\int\frac{dx}{z}\right).
\end{equation}
\item for $c_{w}=-q_{0}^{2}<0$
\begin{equation}
\label{zFundSolHyp}
z^{(1)}(x)=z,\, z^{(2)}(x)=z\sinh\left(2q_{0}\int\frac{dx}{z}\right),\, z^{(3)}(x)=z\cosh\left(2q_{0}\int\frac{dx}{z}\right).
\end{equation}
\item for $c_{w}=0$
\begin{equation}
\label{zFundSolPar}
z^{(1)}(x)=z,\, z^{(2)}(x)=z\left(\int\frac{dx}{z}\right)^{2},\, z^{(3)}(x)=z\int\frac{dx}{z}.
\end{equation}
\end{itemize}
\end{theorem}
Theorem~\ref{thm1} gives us a method of constructing integrable potentials by the following way.
\begin{enumerate}
\item Given a pair $(z,w)$
\item Get a fundamental solution to (\ref{SchEq}) by means of (\ref{FundSolEll}), or (\ref{FundSolHyp}), or (\ref{FundSolPar})
\item Get a fundamental solution to (\ref{lieZ}) by means of (\ref{zFundSolEll}), or (\ref{zFundSolHyp}), or (\ref{zFundSolPar})
\item Get a three-parametric family of integrable potentials (for example, in elliptic case)
\begin{equation*}
\widehat{w}=w+\frac{\widehat{c}}{z^{2}}\left(\alpha_{1}+\alpha_{2}\sin\left(2q_{0}\int\frac{dx}{z}\right)+\alpha_{3}\cos\left(2q_{0}\int\frac{dx}{z}\right)\right)^{-2}
\end{equation*}
with new symmetries
\begin{equation*}
\widehat{z}=z\left(\alpha_{1}+\alpha_{2}\sin\left(2q_{0}\int\frac{dx}{z}\right)+\alpha_{3}\cos\left(2q_{0}\int\frac{dx}{z}\right)\right),
\end{equation*}
where $\alpha_{1}$, $\alpha_{2}$, $\alpha_{3}$ are constants.
\item Again, have a pair $(\widehat{z},\widehat{w})$ and go to step 1.
\end{enumerate}
It is worth to mention that symmetry $z(x)$ not only allows to get solutions to the Schr\"{o}dinger equation, but also produces an infinite hierarchy of integrable potentials and solutions for them.
\subsubsection{Eigenvalue problem}
Consider the Schr\"{o}dinger equation with potential $w(x)-\lambda$ and let
\begin{equation}
\label{dir-cond}
y(a)=y(b)=0,\quad a,b\in\mathbb{R}
\end{equation}
be the Dirichlet boundary conditions.
\begin{theorem}
\label{thm-eig}
Let potentials $w(x)-\lambda$ be integrable and let $z(x,\lambda)$ be their non-trivial symmetries. Then, eigenvalues $\lambda$ of Dirichlet boundary problem (\ref{dir-cond}) for Schrodinger equation (\ref{SchEq}) are solutions of the equation
\begin{equation*}
y^{(1)}(a,\lambda)y^{(2)}(b,\lambda)-y^{(1)}(b,\lambda)y^{(2)}(a,\lambda)=0,
\end{equation*}
where $y^{(1)}$ and $y^{(2)}$ are defined by (\ref{FundSolEll}) in elliptic case, by (\ref{FundSolHyp}) in hyperbolic case, and by (\ref{FundSolPar}) in parabolic case.
\end{theorem}
\begin{proof}
If $z(x,\lambda)$ is a non-trivial symmetry for potential $w(x)-\lambda$, then using theorem~\ref{thm1} we get a general solution to (\ref{SchEq}) in the form
\begin{equation*}
y(x)=C_{1}y^{(1)}(x,\lambda)+C_{2}y^{(2)}(x,\lambda),
\end{equation*}
where $y^{(1)}(x,\lambda)$ and $y^{(2)}(x,\lambda)$ are defined by means of (\ref{FundSolEll}) in elliptic case, by (\ref{FundSolHyp}) in hyperbolic case, and by (\ref{FundSolPar}) in parabolic case. Boundary conditions (\ref{dir-cond}) lead us to the homogeneous linear system for $C_{1}$ and $C_{2}$:
\begin{equation*}
\begin{cases}
  C_{1}y^{(1)}(a,\lambda)+C_{2}y^{(2)}(a,\lambda)=0,
   \\
  C_{1}y^{(1)}(b,\lambda)+C_{2}y^{(2)}(b,\lambda)=0.
 \end{cases}
\end{equation*}
Non-trivial solutions exist if the determinant of this system is equal to zero:
\begin{equation*}
y^{(1)}(a,\lambda)y^{(2)}(b,\lambda)-y^{(1)}(b,\lambda)y^{(2)}(a,\lambda)=0.
\end{equation*}
\end{proof}

Applying the algorithm of generating integrable potentials described above and using results of theorem \ref{thm-eig}, we get a series of potentials for which the eigenvalue problem admits explicit solution.

\begin{example}[Mexican hat]
Consider the eigenvalue problem for the so-called Mexican hat potential:
\begin{equation*}
w(x)=\frac{9\nu^{6}}{4}x^{4}-3\delta x^{2},
\end{equation*}
where $\nu$ and $\delta$ are positive constants. Solving Lie equation (\ref{lieZ}) we get that the symmetry function $z(x,\lambda)$ for this case is
\begin{equation}
\label{symhat}
z(x,\lambda)=\left|\mathrm{HeunT}\left(\frac{\lambda \nu^{6}+\delta^{2}}{\nu^{8}},0,\frac{2\delta}{\nu^{4}},i\nu x\right)\right|^{2},
\end{equation}
where $\mathrm{HeunT}(\alpha,\beta,\gamma,\mathrm{z})$ is the Heun triconfluent function~\cite{SL}, $\mathrm{z}\in\mathbb{C}$, $i$ is an imaginary unit.

Since $\mathrm{L}(w)=\mathbf{S}^{2}(\mathrm{Sol}(w))$, the function $z\in\mathrm{L}(w)$ can be considered as the probability density of the particle for potential  $w(x)$. The probability density $z(x)$ and Mexican hat potential are shown in figure~\ref{mexhat}.

\begin{figure}[h!]
\centering
\includegraphics[scale=.35]{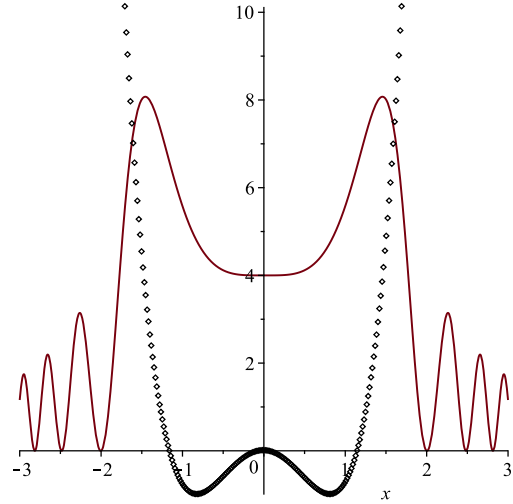}
\caption{Mexican hat potential (point style) and probability density (line style) for the particle with zero energy, $\lambda=0$}
\label{mexhat}       
\end{figure}

Let us now consider the spectral problem for boundary conditions
\begin{equation*}
y(-2)=y(2)=0.
\end{equation*}

Due to theorems \ref{thm1} and \ref{thm-eig}, we get the following equation for $\lambda$:
\begin{equation*}
\mathrm{Im}\left(\exp\left(\frac{-4i(2\nu^{6}-\delta)}{\nu^{3}}\right)\mathrm{HeunT}^{2}\left(\frac{\lambda\nu^{6}+\delta^{2}}{\nu^{8}},0,\frac{2\delta}{\nu^{4}},-2i\nu\right)\right)=0,
\end{equation*}
where $\mathrm{Im}$ is an imaginary part, and the graph of its left-hand side in case of $\nu=\delta=1$ is shown in figure \ref{spec}. One can see that only negative eigenvalues are possible.

\begin{figure}[h!]
\centering
\includegraphics[scale=.35]{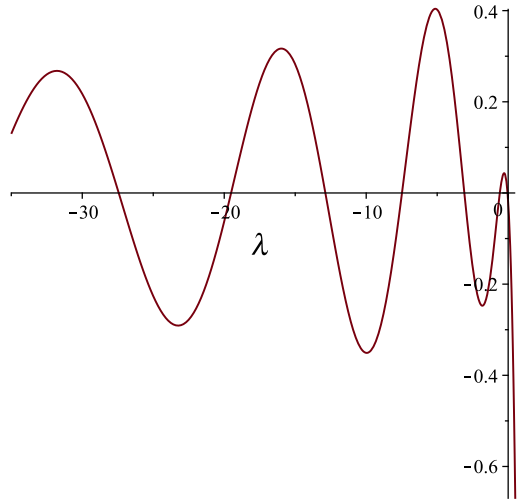}
\caption{Eigenvalues}
\label{spec}       
\end{figure}

Distributions of the density for corresponding eigenvalues $\lambda_{n}$ are shown in figure \ref{dens-mex}.

\begin{figure}[ht!]
\centering
\subfigure[]{\includegraphics[width=0.4\linewidth]{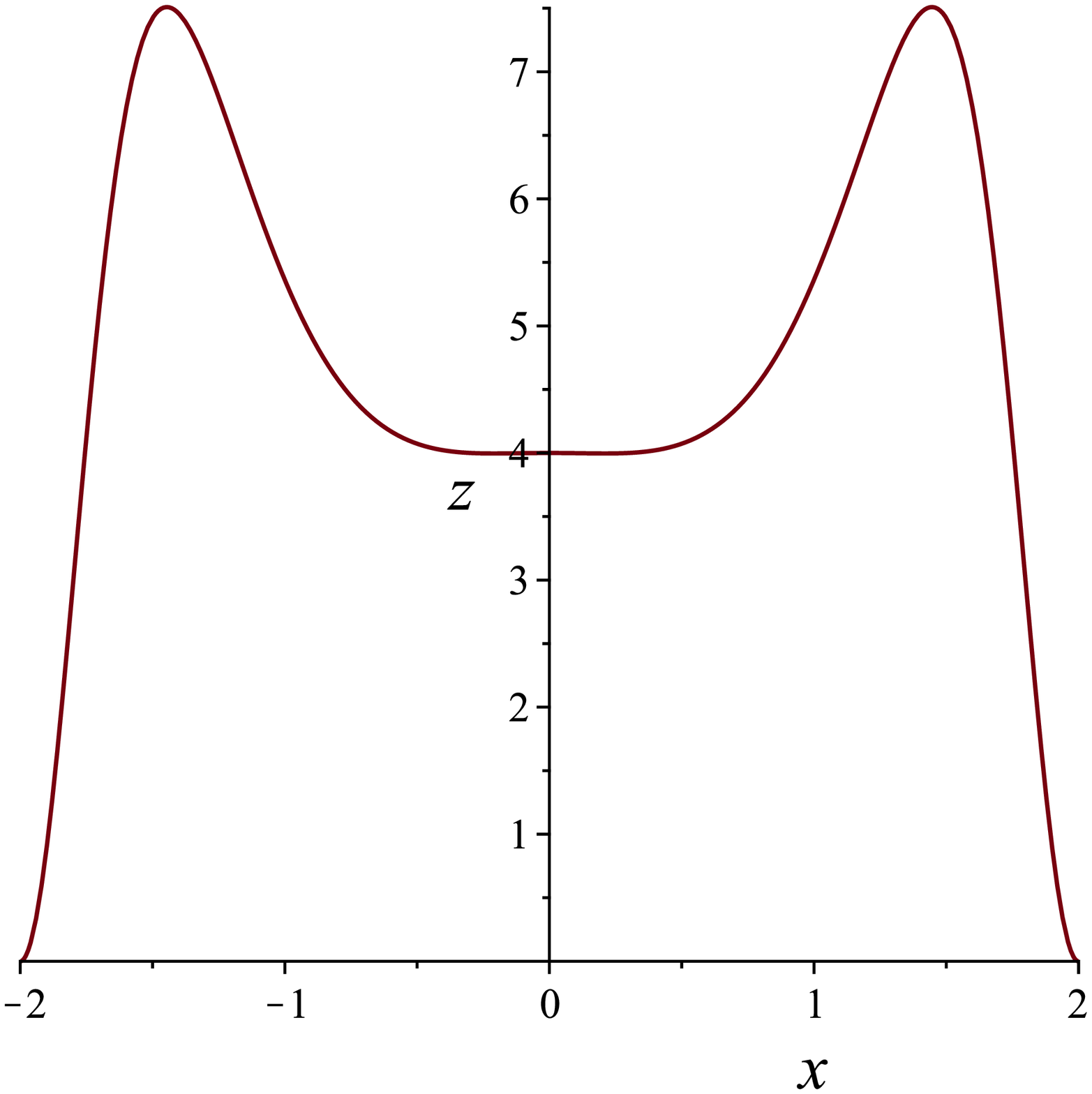} \label{lam1} }
\hspace{4ex}
\subfigure[]{ \includegraphics[width=0.4\linewidth]{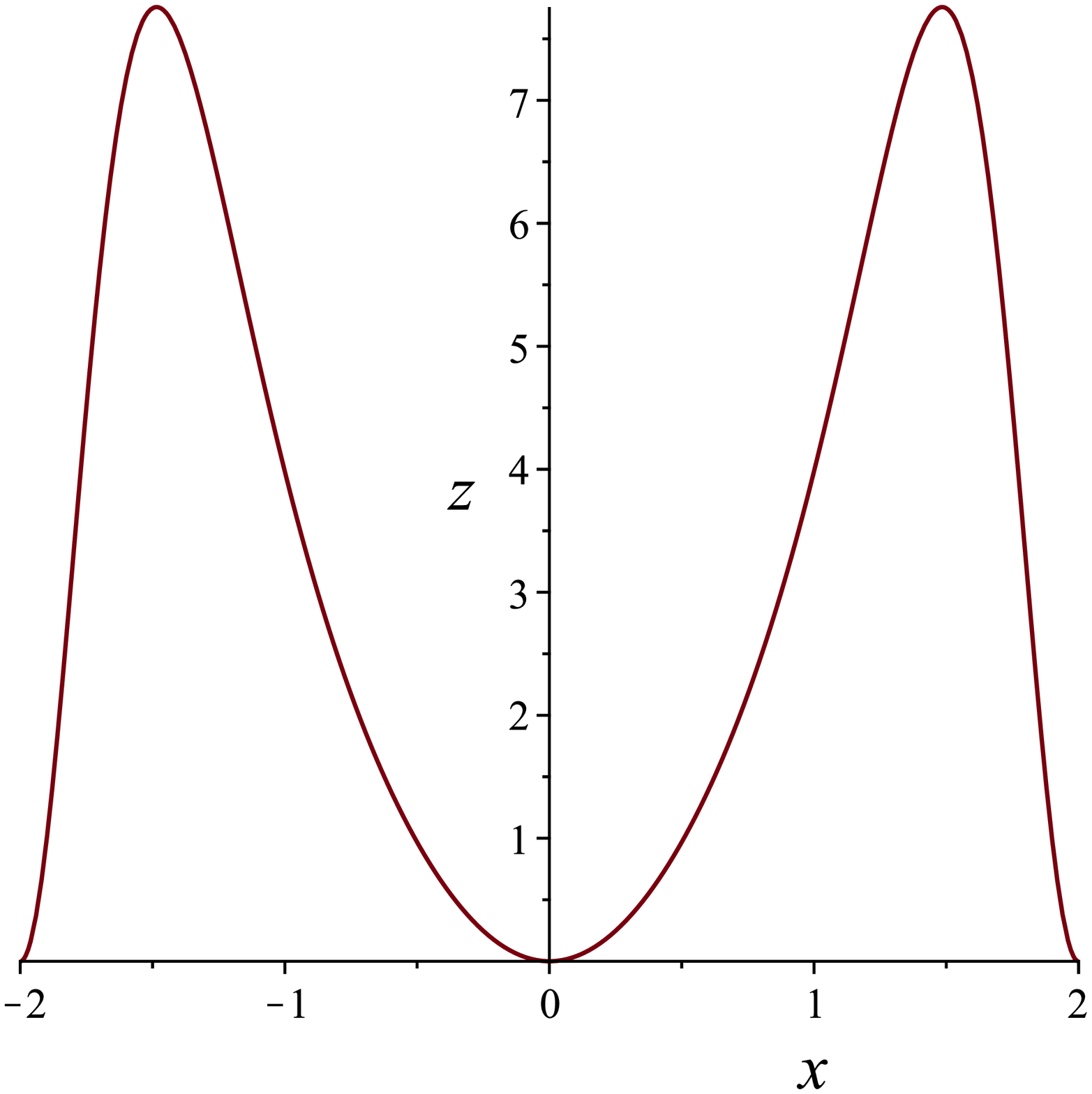} \label{lam2} }
\hspace{4ex}
\subfigure[]{ \includegraphics[width=0.4\linewidth]{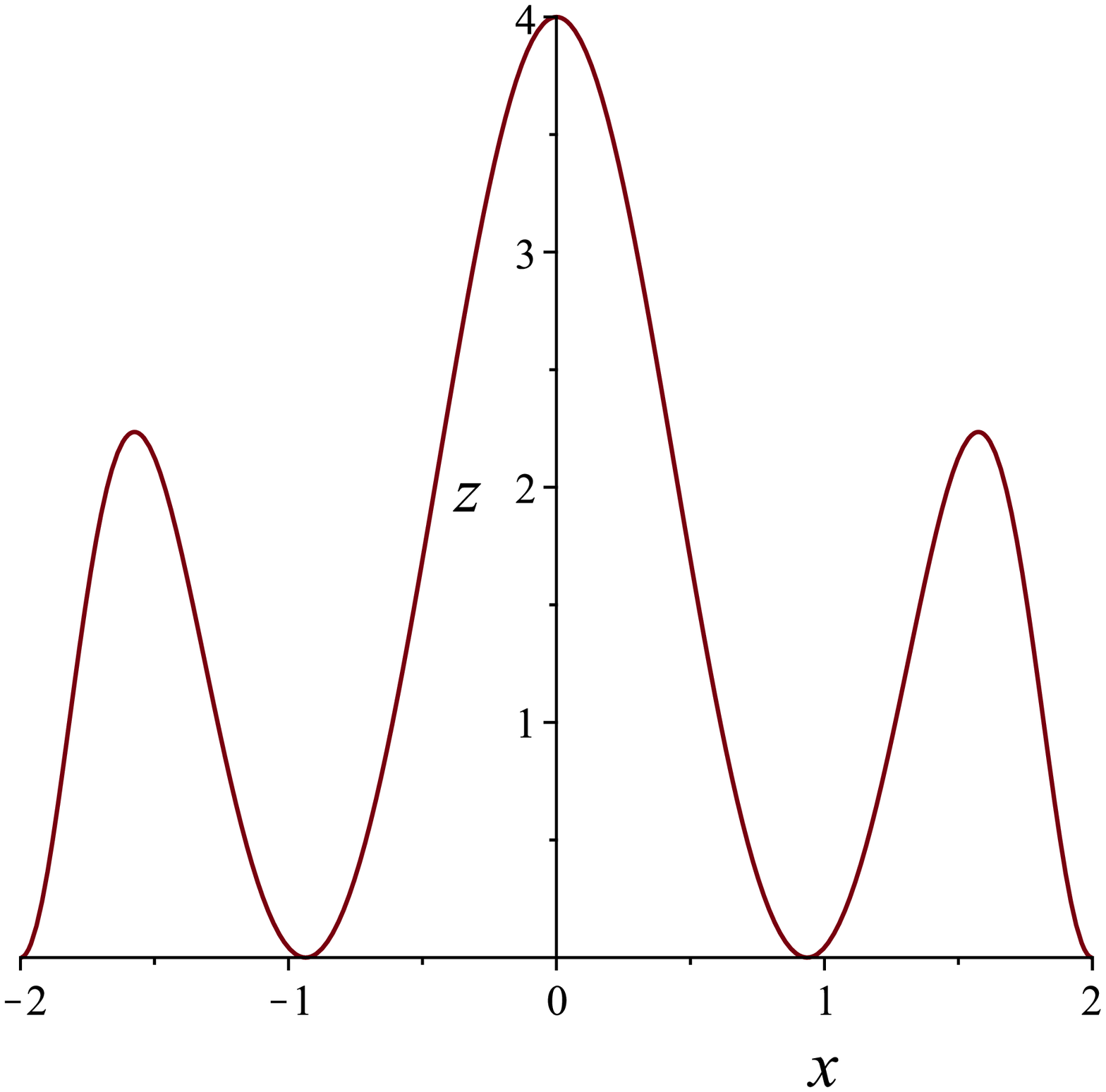} \label{lam3} }
\hspace{4ex}
\subfigure[]{ \includegraphics[width=0.4\linewidth]{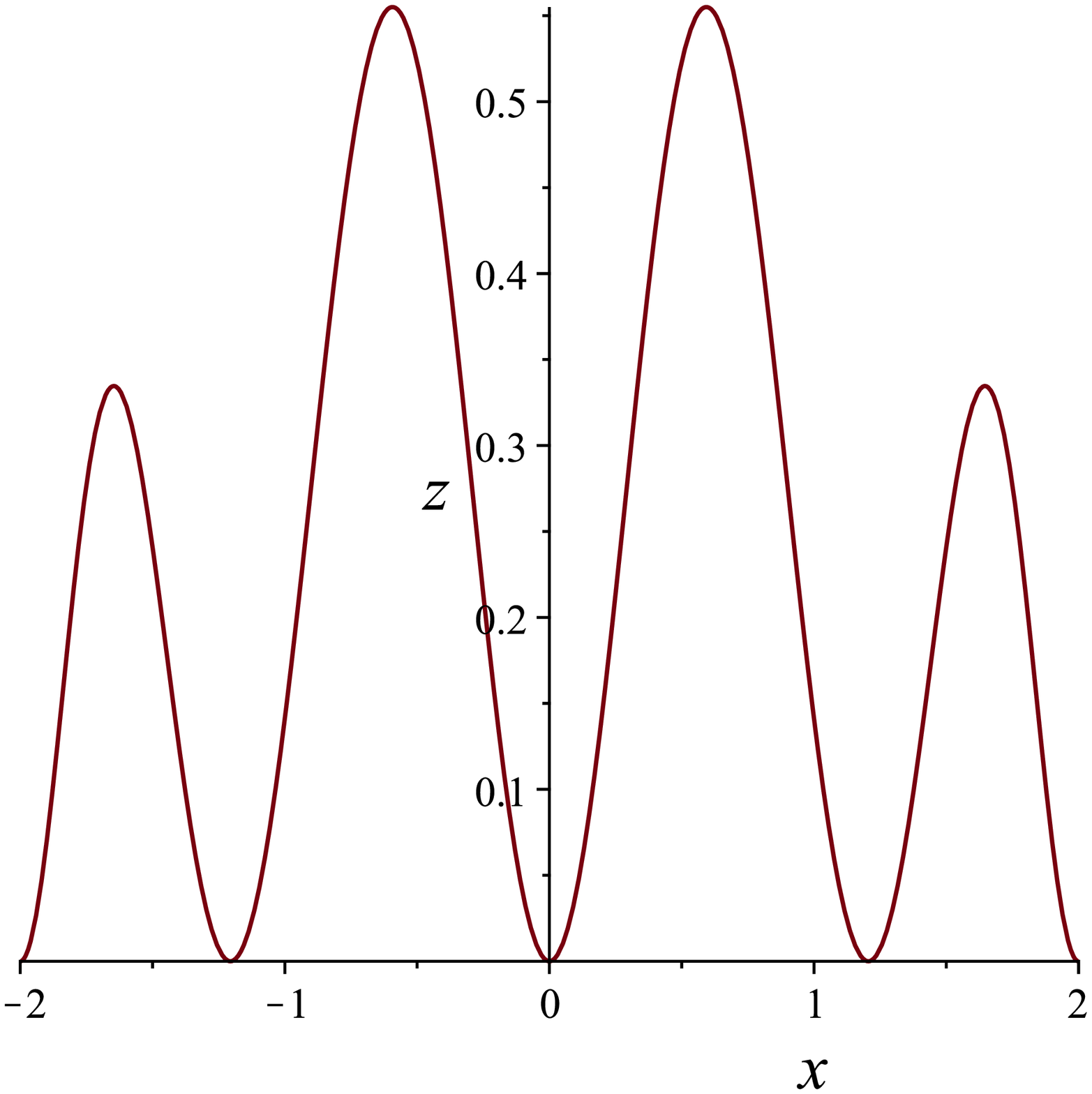} \label{lam4} }
\caption{\footnotesize{Density distributions for Mexican hat potential: \subref{lam1} for $\lambda_{1}=-0.0433$, \subref{lam2} for $\lambda_{2}=-0.59$, \subref{lam3} for $\lambda_{3}=-3.087$, \subref{lam4} for $\lambda_{4}=-7.48$}}
\label{dens-mex}
\end{figure}

The next integrable potential is, for example,
\begin{equation}
\label{newhat}
\widehat{w}=\frac{9\nu^{6}}{4}x^{4}-3\delta x^{2}+\left|\mathrm{HeunT}\left(\frac{\lambda \nu^{6}+\delta^{2}}{\nu^{8}},0,\frac{2\delta}{\nu^{4}},i\nu x\right)\right|^{-4}-\lambda
\end{equation}
with a symmetry given by (\ref{symhat}).

Thus, using the symmetry function $z(x)$ we get solutions of the eigenvalue problem not only for Mexican hat potential, but also for potential (\ref{newhat}) in quadratures.
\end{example}

\section{Lam\'{e} equations}
In this section, we study equations of type (\ref{SchEq}) on elliptic curves:
\begin{equation}
\label{lamepot}
w(x)=C(p_{0})+p_{1}E(p_{0}),
\end{equation}
where $p_{0}=\wp(x)$, $p_{1}=\wp^{\prime}(x)$, $C(p_{0})$ and $E(p_{0})$ are polynomials and $\wp(x)$ is the Weierstrass $p$-function. Weierstrass elliptic functions satisfy an ODE
\begin{equation}
\label{wp}
p_{1}^{2}=4p_{0}^{3}-g_{2}p_{0}-g_{3},
\end{equation}
where $g_{2}$ and $g_{3}$ are \textit{invariants}.

We will look for a symmetry $z(x)$ in the same form as potential:
\begin{equation}
\label{lamesym}
z(x)=A(p_{0})+p_{1}B(p_{0}),
\end{equation}
where $A(p_{0})$ and $B(p_{0})$ are polynomials.

Note that $\wp(x)$ is an even function and therefore due to (\ref{potsym}) three cases are possible.
\begin{itemize}
\item Even case

Here, we have both potential and symmetry as even functions, i.e.
\begin{equation*}
E(p_{0})=0,\quad B(p_{0})=0.
\end{equation*}
\item Odd case

In this case, potential is even, while the symmetry is odd, i.e.
\begin{equation*}
E(p_{0})=0,\quad A(p_{0})=0.
\end{equation*}
\item General case

Both potential and symmetry are neither odd, nor even.
\end{itemize}
Substituting (\ref{lamepot}) and (\ref{lamesym}) to (\ref{lieZ}) and using relation (\ref{wp}) we get an equation
\begin{equation}
\label{genLame}
R_{1}(p_{0})+p_{1}R_{2}(p_{0})=0,
\end{equation}
where
\begin{equation*}
\begin{split}
R_{1}(p_{0})&=\left(4p_{0}^{3}-g_{2}p_{0}-g_{3}\right)^{2}B^{\prime\prime\prime}+3\left(4p_{0}^{3}-g_{2}p_{0}-g_{3}\right)(12p_{0}^{2}-g_{2})B^{\prime\prime}+{}\\&+\left(4(4p_{0}^{3}-g_{2}p_{0}-g_{3})C+300p_{0}^{4}-66p_{0}^{2}g_{2}-48p_{0}g_{3}+\frac{3}{4}g_{2}^{2}\right)B^{\prime}+{}\\&+2B\left((12p_{0}^{2}-g_{2})C+60p_{0}^{3}-9g_{2}p_{0}-6g_{3}\right)+EA(12p_{0}^{2}-g_{2})+{}\\&+2(2EA^{\prime}+BC^{\prime}+AE^{\prime})\left(4p_{0}^{3}-g_{2}p_{0}-g_{3}\right),
\end{split}
\end{equation*}
and
\begin{equation*}
\begin{split}
R_{2}(p_{0})=&(4p_{0}^{3}-g_{2}p_{0}-g_{3})A^{\prime\prime\prime}+\left(18p_{0}^{2}-\frac{3g_{2}}{2}\right)A^{\prime\prime}+4(C+3p_{0})A^{\prime}+{}\\& +2(4p_{0}^{3}-g_{2}p_{0}-g_{3})(BE^{\prime}+2EB^{\prime})+2AC^{\prime}+3BE(12p_{0}^{2}-g_{2}).
\end{split}
\end{equation*}
Equation (\ref{genLame}) is equivalent to the system
\begin{equation}
\label{sysLame}
R_{1}(p_{0})=0,\quad R_{2}(p_{0})=0.
\end{equation}
\subsection{Even case, $E(p_{0})=0$, $B(p_{0})=0$}
In this case, the first equation in (\ref{sysLame}) is trivial and the second one is of the form
\begin{equation}
\label{lameEven}
(4p_{0}^{3}-g_{2}p_{0}-g_{3})A^{\prime\prime\prime}+\left(18p_{0}^{2}-\frac{3g_{2}}{2}\right)A^{\prime\prime}+4(3p_{0}+C)A^{\prime}+2AC^{\prime}=0,
\end{equation}

Let $n$ and $m$ be degrees of the polynomials $A(p_{0})$ and $C(p_{0})$ respectively, i.e. $A(p_{0})=\sum\limits_{i=0}^{n}a_{i}p_{0}^{i}$, $C(p_{0})=\sum\limits_{i=0}^{m}c_{i}p_{0}^{i}$, where $a_{n}\ne0$, $c_{m}\ne0$. Then, the left-hand side of (\ref{lameEven}) is a polynomial in $p_{0}$ of degree $\max(n,n+m-1)$. If $m\ge 2$, then we obtain
\begin{equation*}
(4n+2m)c_{m}a_{n}=0,
\end{equation*}
and we get a contradictory. Therefore only cases $m=1$ make sense. Taking
\begin{equation*}
C(p_{0})=c_{1}p_{0}+c_{0},\quad A(p_{0})=\sum\limits_{i=0}^{n}a_{i}p_{0}^{i},
\end{equation*}
where we put $a_{n}=1$ since the symmetry is defined up to a multiplicative constant, and collecting terms in $p_{0}$, we get the first equation in the form
\begin{equation*}
4(n^{2}+n+c_{1})\left(n+\frac{1}{2}\right)=0,
\end{equation*}
and hence
\begin{equation*}
c_{1}=-n(n+1).
\end{equation*}
The next equations give us coefficients $a_{i}$ consistently.
\begin{theorem}
Coefficients $a_{i}$ are given by the following relations
\begin{equation*}
 \begin{cases}
   \displaystyle a_{n-1}=\frac{c_{0}}{2n-1},
   \\
   \\
     \displaystyle a_{n-2}=\frac{\left(8c_{0}^{2}-ng_{2}(2n-1)^{2}\right)(n-1)}{8(2n-3)(2n-1)^{2}},
 \end{cases}
\end{equation*}
for $i=\overline{n-3,0}$
\begin{equation*}
 a_{i}=\frac{\left((2i^{2}+10i+12)a_{i+3}g_{3}+(2i^{2}+7i+6)a_{i+2}g_{2}-8c_{0}a_{i+1}\right)(i+1)}{4(i+n+1)(2i+1)(i-n)},
\end{equation*}
and $c_{0}$, $g_{2}$, $g_{3}$ may be arbitrary.
\end{theorem}

The possibility for $c_{0}$ to be arbitrary is of great importance for is, because this fact allows us to get solutions to the eigenvalue problem for the Schr\"{o}dinger operator explicitly.
\begin{example}[$n=1$]
We start with $n=1$. In this case we have
\begin{equation*}
w(x)=-2\wp(x)+c_{0},\quad z(x)=\wp(x)+c_{0},
\end{equation*}
which is the classical Lam\'{e} case. Constant $c_{w}$ is found from (\ref{potsym}):
\begin{equation*}
c_{w}=c_{0}^{3}-\frac{c_{0}g_{2}-g_{3}}{4}.
\end{equation*}
The potential and the density distribution are shown in figure~\ref{even-1}.
\begin{figure}[h!]
\centering
\includegraphics[scale=.35]{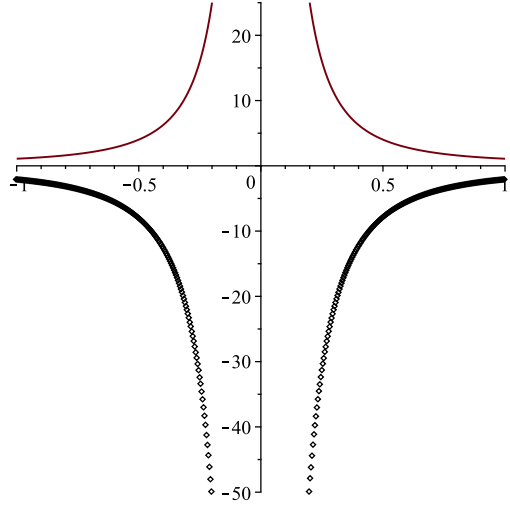}
\caption{Potential (point style) and probability density (line style) for $n=1$}
\label{even-1}       
\end{figure}

Computing integral $\int z^{-1}(x)dx$, we get (see also \cite{LychLych})
\begin{equation*}
\int\frac{dx}{z(x)}=\frac{1}{\sqrt{-c_{w}}}\left(x\wp_{\zeta}(\alpha)+\frac{1}{2}\ln\left(\frac{\wp_{\sigma}(x-\alpha)}{\wp_{\sigma}(x+\alpha)}\right)\right),
\end{equation*}
where $\wp_{\zeta}(x)$ and $\wp_{\sigma}(x)$ are Weierstrass $\zeta$- and $\sigma$-functions respectively and $\alpha$ is the root of the equation
\begin{equation*}
\wp(\alpha)+c_{0}=0.
\end{equation*}
And therefore solution to the Lam\'{e} equation in both cases $c_{w}>0$ and $c_{w}<0$ is given by the same formula
\begin{equation*}
y(x)=\sqrt{|z|}\left(D_{1}\sinh(\mu(x))+D_{2}\cosh(\mu(x))\right),
\end{equation*}
where
\begin{equation*}
\mu(x)=x\wp_{\zeta}(\alpha)+\frac{1}{2}\ln\left(\frac{\wp_{\sigma}(x-\alpha)}{\wp_{\sigma}(x+\alpha)}\right),
\end{equation*}
and $D_{i}\in\mathbb{C}$.

Equation for eigenvalues $c_{0}$ in case of the Dirichlet conditions (\ref{dir-cond}) is of the form
\begin{equation*}
\sinh(\mu(a))\cosh(\mu(b))-\sinh(\mu(b))\cosh(\mu(a))=0.
\end{equation*}
\end{example}
\begin{example}[$n=2$]
Potential and symmetry are
\begin{equation*}
w(x)=-6\wp(x)+c_{0},\quad z(x)=\wp^{2}(x)+\frac{c_{0}}{3}\wp(x)+\frac{c_{0}^{2}}{9}-\frac{g_{2}}{4},
\end{equation*}
The corresponding constant $c_{w}$ is
\begin{equation*}
c_{w}=\frac{1}{324}(c_{0}^{2}-3g_{2})(4c_{0}^{3}-9c_{0}g_{2}-27g_{3}).
\end{equation*}
\end{example}
\begin{example}[$n=3$]
Potential and symmetry are
\begin{eqnarray*}
w(x)&=&-12\wp(x)+c_{0},\\z(x)&=&\wp^{3}(x)+\frac{c_{0}}{5}\wp^{2}(x)+\left(\frac{2c_{0}^{2}}{75}-\frac{g_{2}}{4}\right)\wp(x)+\frac{c_{0}^{3}}{225}-\frac{c_{0}g_{2}}{15}-\frac{g_{3}}{4},
\end{eqnarray*}
The corresponding constant $c_{w}$ is
\begin{equation*}
c_{w}=\frac{c_{0}^{7}}{50625}-\frac{7g_{2}c_{0}^{5}}{11250}-\frac{11g_{3}c_{0}^{4}}{3750}+\frac{31g_{2}^{2}c_{0}^{3}}{6000}+\frac{9g_{2}g_{3}c_{0}^{2}}{200}+\frac{(27g_{3}^{2}-g_{2}^{3})c_{0}}{240}.
\end{equation*}
\end{example}
\subsection{Odd case, $E(p_{0})=0$, $A(p_{0})=0$}
Here, we get the first equation in (\ref{sysLame}) as
\begin{equation}
\label{lameOdd}
\begin{split}
&\left(4p_{0}^{3}-g_{2}p_{0}-g_{3}\right)^{2}B^{\prime\prime\prime}+3\left(4p_{0}^{3}-g_{2}p_{0}-g_{3}\right)(12p_{0}^{2}-g_{2})B^{\prime\prime}+{}\\&+\left(4(4p_{0}^{3}-g_{2}p_{0}-g_{3})C+300p_{0}^{4}-66p_{0}^{2}g_{2}-48p_{0}g_{3}+\frac{3}{4}g_{2}^{2}\right)B^{\prime}+{}\\&+2B\left((12p_{0}^{2}-g_{2})C+60p_{0}^{3}-9g_{2}p_{0}-6g_{3}\right)+{}\\&+2BC^{\prime}\left(4p_{0}^{3}-g_{2}p_{0}-g_{3}\right)=0,
\end{split}
\end{equation}
and the second one is trivial.

If $B(p_{0})$ and $C(p_{0})$ are assumed to be polynomials of degrees $n$ and $m$ respectively, then the left-hand side of (\ref{lameOdd}) is a polynomial of degree $\max(n+m+2,n+3)$ and by the same reasons as in even case only $m=1$ makes sense. If $b_{n}=1$ and $C(p_{0})=c_{1}p_{0}+c_{0}$ then we have a system of $(n+4)$ equations for $(n+4)$ unknowns including $n$ coefficients of $B(p_{0})$, 2 coefficients of $C(p_{0})$ and $g_{2}$, $g_{3}$. The first equation of this system
\begin{equation*}
16(n+2)\left(n^{2}+4n+c_{1}+\frac{15}{4}\right)=0
\end{equation*}
implies
\begin{equation}
\label{c1}
c_{1}=-\frac{15}{4}-n(n+4).
\end{equation}
\begin{theorem}
Coefficients $b_{i}$ are found from a recurrent relation
\begin{equation}
\label{rec}
\begin{split}
&16\,b_{{i-3}} \left( i-3 \right)  \left( i-4 \right)  \left( i-5 \right) +144\,b_{{i-3}} \left( i-3 \right)  \left( i-4 \right)- \\{}&
-8\,g_{{2}} b_{{i-1}} \left( i-1 \right)   \left( i-2 \right)\left( i-3 \right) + \left( 16\,c_{{1}}+300 \right) b_{{i-3}} \left( i-3 \right)- \\{}&
-44\,g_{{3}}b_{{i}}i \left( i-1 \right)  \left( i-2 \right) -96\,g_{{2}}b_{{i-1}} \left( i-1 \right) \left( i-2 \right)  +8\, \left( 4\,c_{{1}}+15 \right) b_{{i-3}}+\\{}&
+32\,c_{{0}}b_{{i-2}} \left( i-2 \right) +8\, \left(4c_{1}+15 \right)g_{2}^{2} b_{{i-3}}b_{{i+1}} \left( i-1 \right)i \left( i+1 \right) +\\{}&
+g_{2}^{2}b_{{i+1}} \left( i-1 \right) i \left( i+1 \right) -36\,g_{{3}}b_{{i}}i \left( i-1 \right) - 4\left( \,c_{{1}}+16 \right) g_{{2}}b_{{i-1}} \left( i-1 \right)+ \\{}&
+24\,c_{{0}}b_{{i-2}}+2\,g_{{3}}g_{{2}}b_{{i+2}}i \left( i+1 \right)  \left( i+2 \right) +3\,g_{2}^{2}b_{{i+1}}i \left( i+1 \right) +\\{}&
+ \left( -4\,c_{{0}}g_{{2}}-4\,g_{{3}} \left( c_{{1}}+12 \right)  \right) b_{{i}}i-2\, \left( 2c_{{1}}+9 \right) g_{{2}}b_{{i-1}}+{}\\&+b_{{i+3}} \left( i+1 \right)  \left( i+2 \right)  \left( i+3 \right) g_{2}^{2}+2\, \left( -c_{{0}}g_{{2}}-g_{{3}} \left( c_{{1}}+6 \right)  \right) b_{{i}}+\\{}&
+3\,g_{{2}}g_{{3}}  b_{{i+2}} \left( i+1 \right)\left( i+2 \right) + \left( -4\,c_{{0}}g_{{3}}+\frac{3g_{2}^{2}}{4} \right) b_{{i+1}} \left( i+1 \right)=0,
\end{split}
\end{equation}
where $i=n+2,n+1,\ldots,3$, with initial conditions $b_{n+5}=b_{n+4}=b_{n+3}=b_{n+2}=b_{n+1}=0$, $b_{n}=1$, and $c_{1}$ is given by (\ref{c1}).
\end{theorem}
The above theorem gives us $b_{j}$, $j=\overline{0,n-1}$, as functions of $c_{0}$, $g_{2}$ and $g_{3}$. The last three equations can be considered as equations for $c_{0}$, $g_{2}$, $g_{3}$ while $b_{j}$ are assumed to be found from (\ref{rec}):
\begin{equation}
\label{gc}
\begin{split}
&6\,b_{{3}}g_{3}^{2}+\left(  \left( 2\,{n}^{2}+8\,n-\frac{9}{2} \right) b_{{0}}-16\,c_{{0}}b_{{1}}+24\,b_{{2}}g_{{2}} \right) g_{{3}}-2\,b_{{0}}c_{{0}}g_{{2}}+\frac{3b_{{1}}g_{2}^{2}}{4}=0,{}\\&
\frac{15}{2}\,g_{2}^{2}b_{{2}}+ \left( -6\,c_{{0}}b_{{1}}+30\,b_{{3}}g_{{3}}+\,b_{{0}} 4\left( {n}^{2}+16\,n-3 \right)  \right) g_{{2}}+{}\\&+6\,g_{{3}} \left( 4\,b_{{4}}g_{{3}}+ \left( {n}^{2}+4\,n-{\frac{25}{4}} \right) b_{{1}}-\frac{4}{3}\,c_{{0}}b_{{2}} \right)=0,{}\\&
{\frac {105\,g_{2}^{2}b_{{3}}}{4}}+\, \left( 84\,b_{{4}}g_{{3}}-10\,c_{{0}}b_{{2}}+4\,b_{{1}} \left(2 {n}^{2}+8\,n-13 \right)  \right) g_{{2}}+{}\\&+60\,b_{{5}}g_{3}^{2}+\left(  \left( 10\,{n}^{2}+40\,n-\frac{285}{2} \right) b_{{2}}-48\,c_{{0}}b_{{3}} \right) g_{{3}}+24\,b_{{0}}c_{{0}}=0.
\end{split}
\end{equation}
One of solutions to (\ref{gc}) is trivial, i.e. $g_{3}=g_{2}=c_{0}=0$. In this case we have $b_{j}=0$, $j=\overline{0,n-1}$ due to (\ref{rec}), and $\wp(x)=(x+w_{0})^{-2}$, where $w_{0}$ is a constant.
\begin{equation}
\label{odd-1}
w(x)=\left(-\frac{15}{4}-n(n+4)\right)(x+w_{0})^{-2},\quad z(x)=-2(x+w_{0})^{-2n-3}.
\end{equation}
One can show that constant $c_{w}$ for pair (\ref{odd-1}) is equal to zero and this case is therefore parabolic. General solution to the corresponding Lam\'{e} equation is
\begin{equation*}
y(x)=\alpha_{1}(x+w_{0})^{-3/2-n}+\alpha_{2}(x+w_{0})^{n+5/2},
\end{equation*}
where $\alpha_{1}$ and $\alpha_{2}$ are constants.
\section*{Acknowledgements}
This work was partially supported by the Russian Foundation for Basic Research (project 18-29-10013).

\end{document}